\documentclass[twocolumn,floatfix,showpacs,pra,aps, letterpaper]{revtex4}

\usepackage[dvips,bookmarks=false]{hyperref}

\usepackage{amsmath}
\usepackage{amssymb}
\usepackage{graphicx}
\usepackage{epsfig}
\usepackage{color}

\newcommand{\ket}[1]{|#1 \rangle}

\newtheorem{theorem}{Theorem}[section]

\newenvironment{proof}[1][Proof]{\begin{trivlist}
\item[\hskip \labelsep {\bfseries #1}]}{\end{trivlist}}

\begin{document}

\title{Message passing in fault tolerant quantum error correction}

\author{Zachary W. E. Evans$^{}\footnote[1]{Electronic address: z.evans@physics.unimelb.edu.au}$ and Ashley M. Stephens}

\affiliation{
Centre for Quantum Computer Technology, School of Physics\\
University of Melbourne, Victoria 3010, Australia.}
\date{\today}

\begin{abstract}
Inspired by Knill's scheme for message passing error detection, here we develop a scheme for message passing error correction for the nine-qubit Bacon-Shor code. We show that for two levels of concatenated error correction, where classical information obtained at the first level is used to help interpret the syndrome at the second level, our scheme will correct all cases with four physical errors. This results in a reduction of the logical failure rate relative to conventional error correction by a factor proportional to the reciprocal of the physical error rate.
\end{abstract}


\maketitle

\section{Introduction}
The effort to design and build a quantum computer is motivated by the discovery of a number of quantum algorithms that are more efficient than their best known classical equivalents \cite{Shor1, Grover1, Aspuru-Guzik1, Farhi1}. Some practical benefit could arise from these algorithms with a quantum computer of only tens or hundreds of qubits, but a quantum computer with thousands of logical qubits will be required to outperform the most powerful classical computers \cite{Meter1}. To mitigate the effects of decoherence and systematic imprecision it is expected that fault tolerant error correction \cite{Shor2, Steane1, Aliferis1} will be required at the expense of additional quantum resources. Motivated by the fact that a limited amount of resources can be allocated to error correction, here we present a new method for error correction which uses message passing to achieve lower error rates than conventional methods that incur the same resource cost.

With the notable exceptions of topological quantum computing \cite{Kitaev1} and the surface code \cite{Bravyi1, Raussendorf1}, most methods for error correction involve concatenation. Our result is in this context and so we briefly review this approach. 

First, information represented by the state of a single physical qubit is encoded in the states of several physical qubits, which together form the logical qubit of an error correction code. Using error correction circuitry constructed from physical operations (including state preparation, measurement, and quantum logic gates such as Hadamard and controlled-NOT) some number of errors affecting the data qubits that make up the logical qubit can be detected and corrected. Typically this is done by interacting the data qubits with some ancillary qubits and then measuring the ancillary qubits to extract a syndrome from which errors are diagnosed. Then, to form a concatenated quantum code, this process is applied recursively - that is, logical qubits are used to encode higher level logical qubits and logical operations are used to construct higher level error correction circuits.

It is known that by concatenating a quantum code, an arbitrarily accurate and arbitrarily large quantum computation can be performed efficiently with faulty components provided that the failure rate of all physical operations is below some threshold \cite{Aharonov1, Aliferis1}. With each level of concatenation the failure rate of the logical operations reduces double-exponentially, however both the number of physical qubits to encode each logical qubit and the time required to perform each logical operation increase exponentially. This increase in resources makes it impractical to use more than just a few levels of concatenation.

Conventional implementations of quantum error correction implicitly assume that errors are equally likely on all data qubits and so the correction is always to apply the fewest possible bit or phase flips required to return the state to the code space. Error correction at each level of concatenation operates independently in this way. However, not only does error correction reduce the probability of a logical error, it also gives some indication of the likelihood that an error has occurred at a higher level of encoding - a logical error is more likely to have occurred if a correction was applied during error correction than if a correction was not applied. If this classical information is passed from lower levels of error correction to higher levels it is possible to correct a larger set of errors. Message passing in quantum computation was first considered in the context of an error detection code, where errors at the level below would become located errors at the next level up and thus be corrected \cite{Knill1}. The idea was also applied to an error correction code in the context of communication across a noisy channel \cite{Poulin1}.

Here we extend this idea by applying message passing to fault tolerant error correction under the nine-qubit Bacon-Shor subsystem code \cite{Bacon1, Bacon2}. In particular we present a method for error correction at the second level of concatenation which uses both the regular syndrome and messages from the first level of concatenation. We refer to this method as message passing error correction (MPEC). Using simulations we find that MPEC will correct all cases with four physical errors, whereas conventional error correction can only guarantee success with up to three physical errors. We also show that MPEC retains fault tolerance and so for physical error rates well below the threshold the use of this method can lead to a significant improvement in the fidelity of a large scale computation relative to conventional error correction.

Note that while message passing error detection has been shown to increase the asymptotic threshold \cite{Knill1, Aliferis2}, our method is designed to improve the performance of error correction without requiring additional resources. We do not explicitly consider any effect that message passing may have on the asymptotic threshold.

\section{Message passing error correction}
\label{sec: method}
In this paper we focus entirely on improving the performance of the second level of concatenated error correction using the [[9,1,3]] Bacon-Shor code. The Bacon-Shor error correction circuits we use are shown in Fig.~\ref{fig: EC circuit} \cite{Aliferis3}.

\begin{figure}
\begin{center}
\resizebox{60mm}{!}{\includegraphics{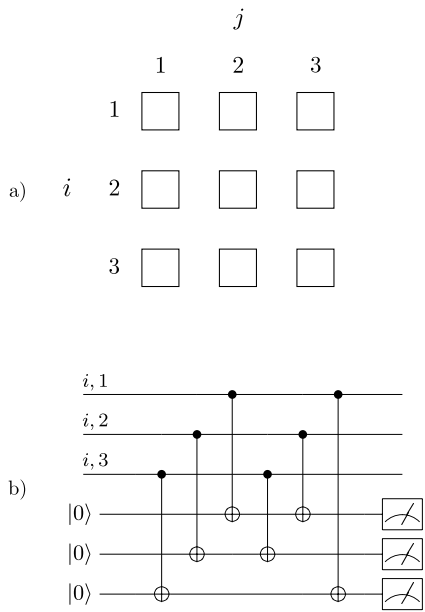}}
\end{center}
\vspace*{-15pt}
\caption{a) The nine data qubits of the Bacon-Shor code can be thought of as the vertices of a 3x3 grid. b) $X$ syndrome extraction circuit \cite{Aliferis3}. The qubit labels indicate that this circuit performs the same operations on each row of data qubits in parallel. $Z$ syndrome extraction has a similar circuit (CNOT gates are reversed and measurement and preparation are in the conjugate basis) which operates on the columns of data qubits.}
\label{fig: EC circuit}
\end{figure}

In the [[9,1,3]] code a logical failure can occur if at least two of the data qubits have errors. When the error correction circuit measures a non-zero syndrome we know that at least one of the data qubits has an error and that any additional error may cause failure. To have a logical error and observe no syndrome requires at least three physical errors. Therefore logical errors are of order $1/p^2$ more likely at logical circuit locations which measured non-zero syndrome. At these level-1 locations, which have detected an error at the physical level, we raise a \emph{flag} to indicate an increased probability of level-1 error. These flags are simply classical information that can be tracked through the quantum circuit.

Bit ($X$) and phase ($Z$) errors can be treated independently. There will be $X$ flags and $Z$ flags each indicating increased probability of the respective errors. Since the $X$ and $Z$ error correction circuits are similar the same rules apply independently to both $X$ and $Z$ correction. Table \ref{tab: outcomes} shows the four possible outcomes of a level-1 error correction block.

\begin{table}
\centering
\begin{tabular}{c|c|c}
\hline\hline
unflagged success  	& 	$US$ & $\mathcal{O}(p^0)$ \\
flagged success 	& 	$FS$ & $\mathcal{O}(p^1)$ \\
flagged failure 		& 	$FF$ & $\mathcal{O}(p^2)$ \\
unflagged falure 	& 	$UF$ & $\mathcal{O}(p^3)$ \\
[1ex]
\hline
\end{tabular}
 \caption{Possible outcomes of a level-1 error correction block. The error correction block is \emph{flagged} if a non-zero syndrome is measured, indicating that at least one error has occurred. Error correction can only detect deviations from the code space, and so logical failures can not be detected without higher level error correction. In terms of the probability of a physical gate failure, $p$, the weights of the outcomes are $0, 1, 2$ and $3$ respectively.}
\label{tab: outcomes}
\end{table}

Using conventional concatenation each level of error correction works independently. The $n^{th}$ level of error correction can correct any single gate failure at the level below. All gates in the error correction circuit are assumed to be equally likely to fail, and so when a non-trivial error syndrome is measured the lowest possible number of errors is assumed and the corresponding correction is applied. Information about the syndrome and the corrections applied is then thrown away. In our method we use this same information to estimate the relative probability of gate failures at a higher level. The key ideas are that more errors can be corrected if the locations of the errors are known and that logical failure is more likely when a non-trivial correction is applied. That is, there is a significant correlation between flags and logical failure.

Codeword states with two errors share the same set of syndromes as codeword states with single errors. However, with the additional flag information, is it possible to distinguish between a larger set of errors. We call the combined information of the flags and the measured syndrome the \emph{super syndrome}. At the second level of error correction, the syndrome of a single $FF$ may be indistinguishable from that of two $FF$s, and so without message passing, four physical errors (if they result in two $FF$s) can cause a logical error at level 2. When the full super syndrome is considered one $FF$ is distinguishable from two $FF$s and so we can attempt to correct two level-1 errors at once. Cases with two $FF$s can be accurately diagnosed and corrected. $1 UF + 2 FS$ may share the same super syndrome with two $FF$, but $1 UF + 2 FS$ is weight $5$ in terms of the physical error rate. It turns out that all weight 4 cases can be corrected using MPEC.

When the error correction circuitry itself is made from noisy gates the problem is significantly more complicated. The error correction circuits are designed to be fault tolerant - a fault at any given physical location will not spread to multiple data qubits in the same logical qubit. However, this type of fault tolerance does not guarantee the protection of the properties that make message passing assisted error correction superior. For example, consider an encoded level-1 CNOT that is part of a level-2 error correction circuit. A single physical error during this CNOT could result in a flag being raised on one or both of the qubits depending on the location of the error. Two physical errors could result in $US+FF$, $US+FS$, $FS+FS$, $FF+FF$, etc. (Fig.~\ref{fig: CNOT errors}). With three physical errors, all combinations of flags and failures are possible. Therefore not all two $FF$ cases are weight 4 for example. This effect complicates the process of interpreting the super syndrome and ensuring fault tolerance. We found that not all fault tolerant circuits for conventional error correction can be readily adapted to MPEC. 

\begin{figure}
\begin{center}
\resizebox{60mm}{!}{\includegraphics{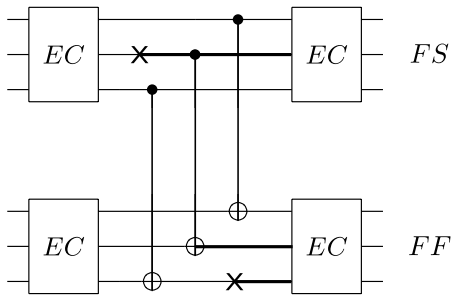}}
\end{center}
\vspace*{-10pt}
\caption{Physical errors during a level-1 two-qubit gate can result in many different combinations of flags and failures on the two logical qubits depending on where the errors occur. In this example there are two physical errors, one of the errors affects both level-1 qubits and the other error only affects the bottom level-1 qubit. Both of the trailing EC boxes will see a non-zero syndrome and so each will raise its flag. In this case the final result will be \em{flagged success} on the top qubit and a \em{flagged failure} on the bottom qubit. In general, each physical error can affect one or both logical qubits. Any combination of $US$, $FS$, $FF$ and $UF$ is possible as long as the error weight one each logical qubit is not more than the number of physical errors.}
\label{fig: CNOT errors}
\end{figure}

To provide the best possible error correction, we aim to apply a correction corresponding to the most likely cause of the super syndrome. We assume a low physical error rate, and so the most likely cause is the one with the fewest physical errors.

Here is an overview of how our scheme works: The error correction circuit used is shown in Fig.~\ref{fig: EC circuit}. $X$ and $Z$ error correction are treated separately. There are $X$ flags and $Z$ flags, which represent potential $X$ and $Z$ errors respectively. A flag is raised on a level-1 qubit whenever lower level error correction measures a non-zero syndrome. Each flag has its own identity. The propagation of each flag is tracked through the circuit as if it were an actual error, spreading through CNOT gates. $X$ flags copy downwards through CNOT gates and $Z$ flags copy upwards. When the second level syndrome is measured we look for combinations of one, two, or three flags on the ancilla which if corresponded to actual errors would produce the syndrome that was measured. Such a combination is called a \emph{flag match}. (For examples of this see figures in Appendix A.) If a flag match is found, we apply corrections to the data qubits that share the same flags as used in the flag match. Some super syndromes have more than one flag match, in these cases we preference the match which implies the fewest errors and \emph{results in the fewest corrections to the data qubits}. If no flag match is found, we fall back to the default error correction action; ie. we correct based on the syndrome alone. After error correction, all flags found on the ancillary qubits at the point of measurement are cleared from the data qubits - these flags have either been established to be $FF$ (and have been corrected) or $FS$, so they are of no further use.

Note that although $3 FS + 1 UF$ has equal weight to $3 FF$, triple flag matches must still be preferenced over the default correction to avoid particular cases in which four physical errors would otherwise cause level-2 failure due to certain combinations of 2-qubit gate failures. (See Appendix A.)

Further levels of concatenated error correction will compound the benefits of MPEC. However, making use of message passing at all levels is not so easy. Finding a flag scheme that works at all levels is non-trivial, and so for this paper we have focused on just two levels of error correction. A simple method for concatenation beyond level 2 is to alternate between error correction that provides flags for the next level, and error correction that uses flags from the previous level but provides no flags. For this alternating sequence the number of physical errors required for logical failure at increasing levels of concatenation is $1, 2, 5, 10, 25, 50, 125$; in contrast with the regular scaling, $1, 2, 4, 8, 16, 32, 64$, for the same physical resources.

\section{Fault tolerance}
\label{sec: fault tolerance}
The procedure described in the previous section ensures that \emph{no combination of four physical errors will result in a level-2 logical error}. This is an improvement over conventional concatenation which can only guarantee the correction of three physical errors at level-2. However, to show that our scheme is effective for large circuits we must show that an arbitrary length chain of level-2 locations will not fail for any combination of four physical errors in each level-2 extended rectangle. An extended rectangle is defined to be an encoded gate with error correction before and after \cite{Aliferis1}.

Usually, fault tolerance of error correction circuits is considered from the point of view that all levels act independently. The basic rule is that a single error should not spread to multiple errors on the same data qubit. For MPEC the fault tolerant condition is more complicated. Since corrections are based on the entire super syndrome we must take into account the error weight of each of the possibilities in Table \ref{tab: outcomes} rather than just counting the failures.

To ensure fault tolerance, an MPEC box must satisfy the following condition:
\begin{equation}
a + b \leq 4 \Rightarrow c \leq b
\label{eq: FT rule}
\end{equation}
where $a$ is the error weight of the state entering the MPEC box, $b$ is the weight of errors occurring inside the box, and $c$ is the error weight of the state leaving the box after any corrections have been applied. Error weights are calculated according to Table \ref{tab: outcomes}. Evidence that the MPEC scheme satisfies the condition can be found in Appendix A. A proof that this condition guarantees the performance of the error correction can be found in Appendix B.

\section{Simulation}
\label{sec: simulation}
To test our message passing error correction we simulate a level-2 CNOT extended rectangle. We compare the failure rates with and without message passing. The failure rate of this circuit is meant to approximate the failure rate of the level-2 CNOT which might be an algorithmic location.

To simulate the error correction circuit we need only simulate the propagation of errors that occur during the circuit \cite{Steane3}. This avoids having to store the complete state of the computer which we assume is an arbitrary codeword state perturbed by the errors. The circuit is deemed to fail if a readout of the data qubits at the end of the circuit in either the $X$ or $Z$ basis would not produce the correct output. Equivalently, a circuit is defined to have succeeded if an errorless error correction cycle applied to its output state would produce the correct state.

For $p>4\times10^{-5}$ we perform many simulations of the entire two level circuit with a fixed error rate and tally the failures. At each physical circuit location we apply an error with probability $p$. For all single qubit locations (preparation, measurement, memory) the error is a randomly selected single qubit Pauli error, $X$, $Y$, or $Z$. For two-qubit locations (CNOT) the error is a randomly selected two-qubit Pauli error ($X\otimes I$, $X\otimes X$, $X\otimes Y$, etc.) Errors at all locations are independent.

For error rates below $4\times10^{-5}$ the number of trials required to generate sufficiently accurate statistics directly is too large; so we use an alternative approach. Instead of giving each circuit location some probability of error, we simulate the circuit with exactly $i$ errors placed randomly. This is repeated many times for $i=[4,12]$ to generate the probability $r_i$ that the circuit fails given that there were $i$ errors. These conditional probabilities can be combined to give the failure rate of the circuit as a function of $p$,
\begin{equation}
p_{(2)} = \sum^N_{i=0}{r_i {N\choose i}p^i(1-p)^{N-i}},
\label{eq: expansion}
\end{equation}
where $N$ is the number of locations in the entire circuit, 72657.

To ensure randomness over the large number of trials required we use the SIMD-oriented Mersenne Twister pseudo random number generator \cite{Saito1}. The results of the simulations are in Fig.~\ref{fig:results}. Note that Eq.~\ref{eq: expansion} is truncated after $i = 12$ after which $r_i$ is set to equal zero. This is why the curves drop to zero at higher error rates. A second set of lines connect the direct data points taken above $p = 4\times10^{-5}$.
\begin{figure}
\begin{center}
\resizebox{88mm}{!}{\includegraphics{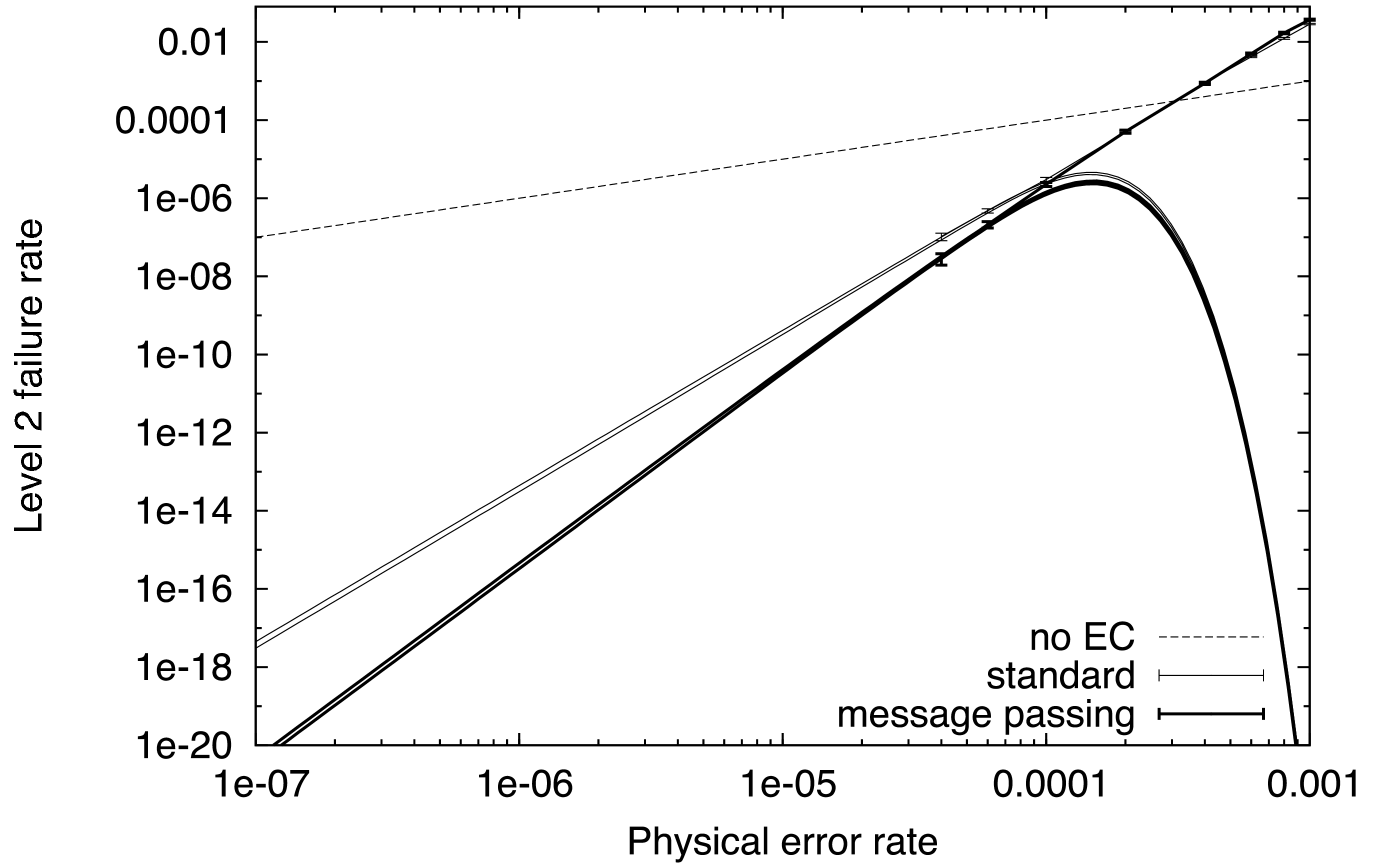}}
\end{center}
\vspace*{-10pt}
\caption{Level-2 logical failure rate vs physical error rate with and without message passing. The data points at high error rates come from direct simulations of a level-2 CNOT extended rectangle, where the error bars indicate $\pm 2\sigma$ statistical error. The lines at low error rates are given by Eq.~\ref{eq: expansion}, where $r_i$ are obtained from simulations for $i\le12$ and set to zero for $i > 12$. To show the statistical error in the parameters $r_i$ there are two lines for both standard and message passing error correction. The top line in each pair is $2\sigma$ above the mean and the bottom is $2\sigma$ below the mean.}
\label{fig:results}
\end{figure}

Our simulations indicate that all combinations of four errors are corrected by MPEC. This is in contrast to standard error correction which fails for some four error combinations. For values of $p$ close to the threshold, $p_{th}$, the failure rate is not significantly reduced, but for $p\ll p_{th}$ the failure rate is reduced by a factor of 
\begin{equation}
\frac{p^{std}_{(2)}}{p^{mp}_{(2)}} \approx \frac{r_4^{std}}{r_5^{mp}}\frac{5}{Np} \approx \frac{10^{-4}}{p}.
\end{equation}
The superscripts indicate standard error correction and message passing error correction.

\section{Conclusions and further work}
In the region of interest - for physical error rates where two or three levels of error correction is sufficient for logical computation - MPEC is significantly better than standard quantum error correction. We emphasize that this benefit is achieved only at the cost of additional classical processing. For more than two levels of error correction the benefit can be compounded, but the problem of finding an optimal and general flagging scheme (one that can be applied at all levels) is an open question and the subject of further work. Also of interest is that the complexity of correctly interpreting the super syndrome appears to depend on the circuits that are used to extract the syndrome. Can MPEC work fault tolerantly for other codes such as the Steane code \cite{Steane1} and if so is it compatible with the most compact Steane circuits \cite{DiVincenzo1}?

\section*{Acknowledgments}
We acknowledge the assistance of Anthony Agius, Geoff Attwater, Magdalena Carrasco, Andrew Greentree, Charles Hill, and Lloyd Hollenberg. We acknowledge the support of the Australian Research Council (ARC), the US National Security Agency (NSA), and the Army Research Office (ARO) under contract number W911NF-04-1-0290.
\bibliographystyle{unsrt}

\begin{thebibliography}{10}

\bibitem{Shor1}
P.~W.~Shor.
\newblock {\em Soc. Ind. App. Math.} \textbf{26}, 1484 (1997).

\bibitem{Grover1}
L.~Grover.
\newblock {\em Proc. IEEE Symp. Th. Comp.} \textbf{28}, 212 (1996).

\bibitem{Aspuru-Guzik1}
A.~Aspuru-Guzik, A.~D.~Dutoi, P.~J.~Love, and M.~Head-Gordon.
\newblock {\em Science} \textbf{309}, 1704 (2005).

\bibitem{Farhi1}
E.~Farhi, J.~Goldstone, and S.~Gutmann
\newblock arXiv:quant-ph/0702144v2 (2007).

\bibitem{Meter1}
R.~Van Meter.
\newblock Ph.D. thesis, Keio University, arXiv:quant-ph/0607065 (2006).

\bibitem{Shor2}
P.~W.~Shor.
\newblock {\em Proc. Ann. Symp. Found. Comp. Sci.} \textbf{37}, 56 (1996).

\bibitem{Steane1}
A.~M.~Steane.
\newblock {\em Phys. Rev. Lett.} \textbf{78}, 2252 (1997).

\bibitem{Aliferis1}
P.~Aliferis, D.~Gottesman, and J.~Preskill.
\newblock {\em Quant. Inf. Comp.} \textbf{6}, 97 (2006).

\bibitem{Kitaev1}
A.~Kitaev.
\newblock {\em Ann. Phys.} \textbf{303}, 2 (2003).

\bibitem{Bravyi1}
S.~Bravyi and A.~Kitaev.
\newblock arXiv:quant-ph/9811052 (1998).

\bibitem{Raussendorf1}
R.~Raussendorf and J.~Harrington. 
\newblock {\em Phys. Rev. Lett.} \textbf{98}, 190504 (2007).

\bibitem{Aharonov1}
D.~Aharonov and M.~Ben-Or.
\newblock {\em Proc. $29^{th}$ Ann. ACM Symp. Th. Comp.}, 176 (1998).

\bibitem{Knill1}
E.~Knill.
\newblock {\em Nature} \textbf{434}, 39 (2005).

\bibitem{Poulin1}
D.~Poulin.
\newblock {\em Phys. Rev. A} \textbf{74}, 052333 (2006);

\bibitem{Bacon1}
D.~Bacon.
\newblock {\em Phys. Rev. A} \textbf{73}, 012340 (2006).

\bibitem{Bacon2}
D.~Bacon and A.~Casaccino.
\newblock {\em Proc. Ann. Alerton Conf.} \textbf{44}, arXiv:quant-ph/0610088 (2006).

\bibitem{Aliferis2}
P.~Aliferis.
\newblock arXiv:0709.3603 (2007).

\bibitem{Aliferis3}
P.~Aliferis and A.~W.~Cross.
\newblock {\em Phys. Rev. Lett.} \textbf{98}, 220502 (2007).

\bibitem{Steane3}
A.~M.~Steane.
\newblock {\em Phys. Rev. A} \textbf{68}, 042322 (2003).

\bibitem{Saito1}
M.~Saito and M.~Matsumoto.
\newblock {\em Monte Carlo and Quasi-Monte Carlo Methods} \textbf{2}, 607 (2006).

\bibitem{DiVincenzo1}
D.~P.~DiVincenzo and P.~Aliferis.
\newblock {\em Phys. Rev. Lett.} \textbf{98}, 020501 (2007).

\end{thebibliography}

\appendix
\section{}
\label{sec: Ap. A}
The primary purpose of this section is to demonstrate that the circuits and rules described in Section \ref{sec: method} satisfy the fault tolerance condition, Eq.~\ref{eq: FT rule}. In doing so we also give examples of flag propagation and flag matching which could be useful for a general understanding of the MPEC scheme.

It is important to note that the [[9,1,3]] Bacon-Shor subsystem code has some symmetries which can be used to greatly simplify the fault tolerance analysis. The nine data qubits form a 3x3 grid for which the rows essentially form a repetition code that protects against $X$ errors and the columns a repetition code that protects against $Z$ errors (Fig.~\ref{fig: EC circuit}). The code is defined to take the binary addition of the columns/rows for the syndrome, so that any even number of $X$ errors on a column is benign and even number of $Z$ errors on a row is benign. The entire error correction circuit consists of three sets of the circuit shown in Fig.~\ref{fig: EC circuit} run in parallel followed by three sets of the equivalent $Z$ syndrome extraction. As $X$ and $Z$ error correction are similar and operate independently, and all errors and flags from circuits performed in parallel are added in binary at end of the syndrome extraction, only a single block (the circuit in Fig.~\ref{fig: EC circuit}) need be considered. All errors and flags on parallel circuit blocks are equivalent to the same errors and flags on a single block. Although these symmetries are used for our explanations and diagrams, our simulations always used the full circuits with both $X$ and $Z$ errors.

CNOT gates copy $X$ errors from the control to the target, and copy $Z$ errors from the target to the control.
\begin{align}
\begin{aligned}
CNOT.X_{ctrl}\ket{\psi} & = & X_{ctrl}.X_{targ}.CNOT\ket{\psi} \\
CNOT.X_{targ}\ket{\psi} & = & X_{targ}.CNOT\ket{\psi} \\
CNOT.Z_{ctrl}\ket{\psi} & = & Z_{ctrl}.CNOT\ket{\psi} \\
CNOT.Z_{targ}\ket{\psi} & = & Z_{ctrl}.Z_{targ}.CNOT\ket{\psi}
\end{aligned}
\end{align}
Flags represent potential errors and so they are propagated in the same way that the errors are. Each flag that is raised is given a unique identity, but they can by copied to multiple qubits by CNOT gates in the same way that errors are. For the diagrams in this section the flags are distinguishable by their patterns.

A set of flags is said to match the syndrome if identifying those flags as errors can exactly describe the measured syndrome. When testing for a flag match, each ancillary qubit has its flags added in binary, flags that are not part of the set to be tested are ignored, if the result of all the ancillary qubits is identical to the syndrome then the set of flags is a match.

The figures shown in this section depict the propagation of flags and errors on the second level of error correction. It is implicit that at every circuit location there is a lower level of error correction being performed which can raise flags. These flags are indicated by the flags drawn on the circuits. Crosses indicate a level-1 error (or correction at the end of the circuit). The propagation of the errors is shown using bold lines. At the end of the circuit, the figures show the values of the syndrome measurements and the flags that have propagated onto the ancilla. The flags used in a flag match are highlighted and corrections are shown by crosses at the end of the circuit.

The flags and crosses result from physical errors during the lower level error correction. The error weights of these locations can be inferred from Table \ref{tab: outcomes}. Errors during the CNOT gates can result in flags and errors on both qubits (Fig.~\ref{fig: CNOT errors}). In these cases the number of physical errors during the CNOT limit the error weight of each logical qubit.

To demonstrate that Eq.~\ref{eq: FT rule} is satisfied (ie. that the error weight of the output from an EC box is no greater than the number of errors occurring inside the box) we reason on a case by case basis with cases grouped according to the weight of the incoming errors.

\subsection{Failures cases, $a + b \ge 5$}
To understand how MPEC succeeds on a case by case basis, it is instructive to first examine the typical failure cases.

Fig.~\ref{fig: MPEC fail1} shows how a single $FS$ can cause failure by masking the flags of two $FF$s. For this failure to occur, the two $FF$s and the $FS$ must each affect different qubit lines (rows for $Z$ errors, columns for $X$ errors).

Fig.~\ref{fig: MPEC fail2} shows that a pair of $FS$s can match the syndrome of a $UF$ and thus cause the wrong correction to be applied. Again, each of the $FS$s and the $UF$ must affect different qubit lines. If there is a flag from an $FS$ that follows the same path as a $UF$, the combination of the $UF$ and the $FS$ is equivalent to an $FF$.

It is worth noting that no number of $FS$s alone can ever cause failure.

\begin{figure}[h!]
\begin{center}
\resizebox{80mm}{!}{\includegraphics{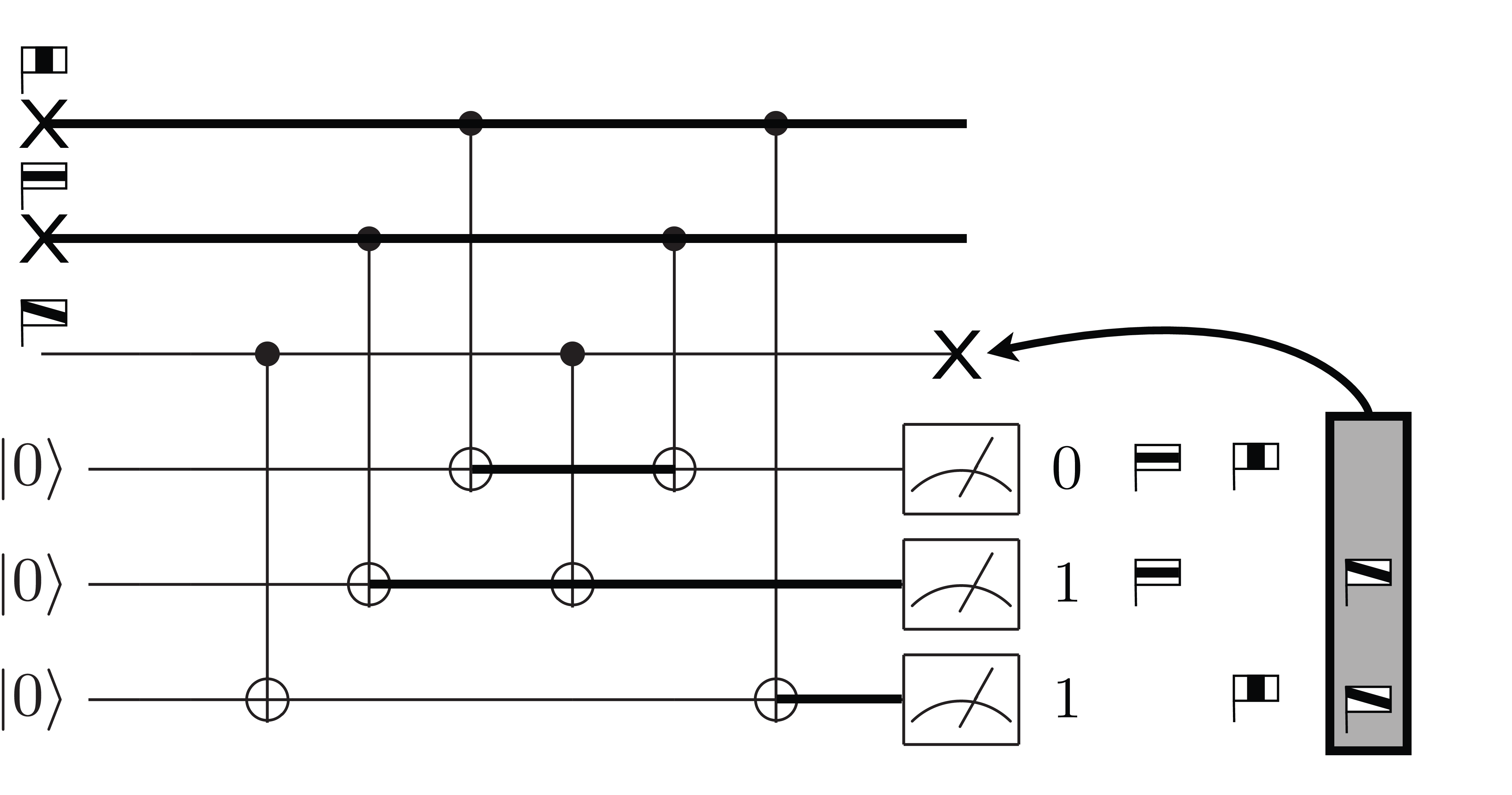}}
\end{center}
\vspace*{-10pt}
\caption{MPEC failure from a error weight 5 case with two incoming $FF$s and one $FS$. In this case the single flag from the $FS$ matches the syndrome and is thus preferred over the pair of flags from the $FF$s. This will result in an incorrect correction and hence logical failure. In this particular example conventional error correction and MPEC would both fail.}
\label{fig: MPEC fail1}
\end{figure}

\begin{figure}[h!]
\begin{center}
\resizebox{80mm}{!}{\includegraphics{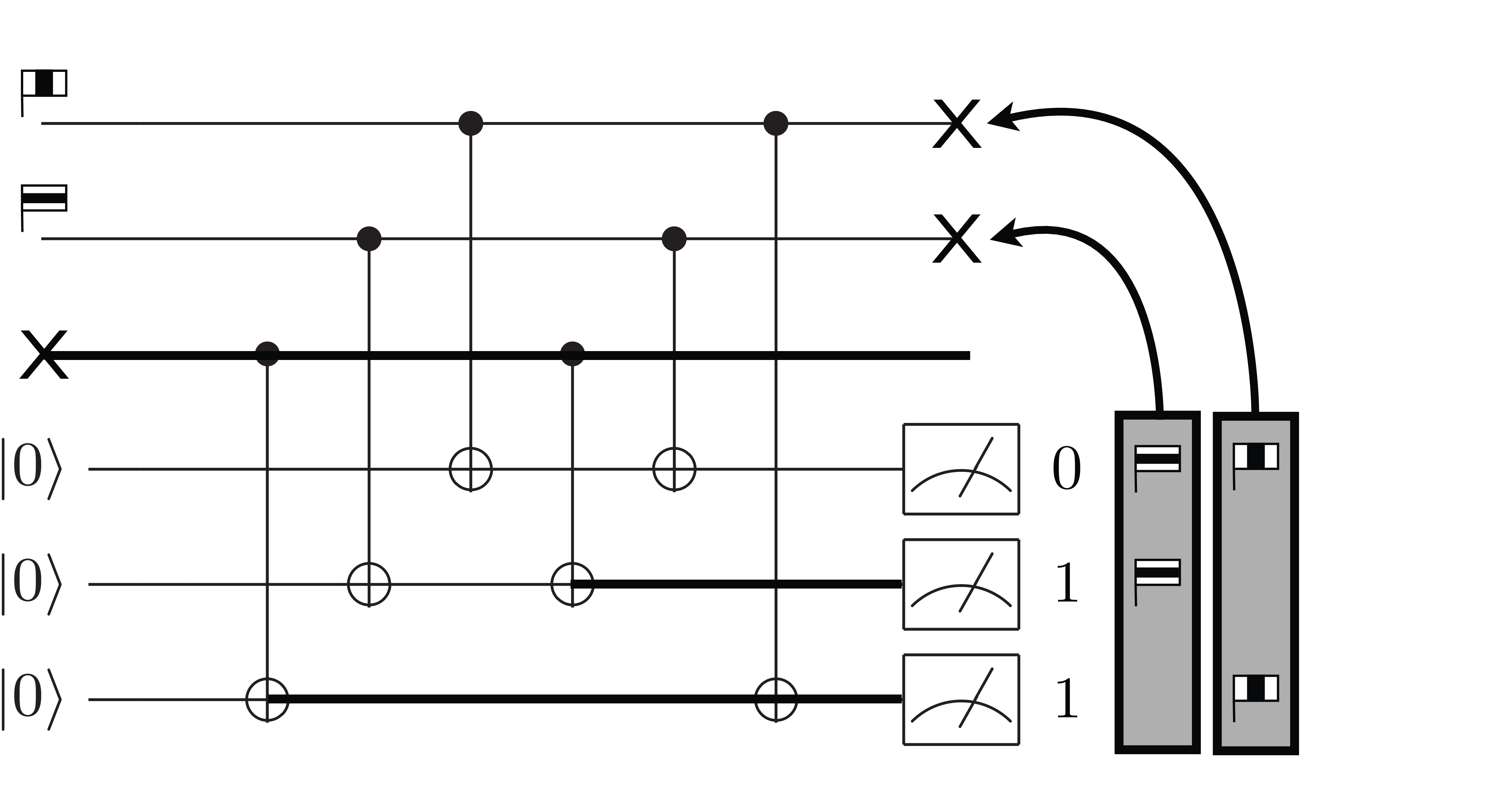}}
\end{center}
\vspace*{-10pt}
\caption{MPEC failure from a error weight 5 case with two incoming $FS$s and one $UF$. The two flags match the syndrome and so are assumed to be errors. This causes MPEC to apply the wrong correction and results in logical failure. This case is of particular interest because it is an example of where conventional error correction would succeed but MPEC fails.}
\label{fig: MPEC fail2}
\end{figure}

\subsection{Cases with $a = 4, b = 0$}
All incoming flags will be passed down to the ancillary qubits and hence, by the rules of MPEC, be cleared from the data after the syndrome measurement.

Any incoming $UF$ will be corrected by conventional error correction rules, as a single $FS$ is not enough to interfere.

Two incoming $FF$s is the exact case that MPEC was designed to deal with. The flags will match the syndrome and both errors will be corrected.

A single incoming $FF$ cannot be masked by one or two $FS$s. The two $FS$s can be positioned to make a two-flag match, but MPEC will always prefer to use the single flag match which is provided by the $FF$.

\subsection{Cases with $a = 3, b \le 1$}
This time, we have the freedom to place a single $FS$ anywhere within the circuit. But this doesn't affect any of arguments given for $a = 4$. An $FS$ cannot interfere with the correction of a single incoming $FF$ or $UF$ no matter where it is placed. One of the important properties of the EC circuit we have chosen is that 
no single $FS$ can mask an error to prevent its correction, because each incoming error on the data will affect two of the ancillary qubits. This property is not a feature of all (conventionally) fault tolerant circuits.

An $FS$ on the data towards the end of the circuit will cause a the $FS$ to be passed out of the EC box.

\subsection{Cases with $a = 2, b \le 2$}
With the up to two errors within the circuit we now have the ability to make an $FF$, or even a pair of $FF$s at a CNOT.

In the worst case there is an incoming $FF$ and a pair of $FF$s that arise from errors during a level-1 CNOT (Fig.~\ref{fig: MPEC2}); to successfully correct this we must use a triple flag match.

CNOT errors can that result in $FF+FS$ pairs in addition to an incoming $FF$ (Fig.~\ref{fig: MPEC3}), but not in any configuration that will cause problems. The $FS$ and $FF$ from CNOT failure can be thought of as an effect on the same data qubit line. With this in mind it is clear that no situation like Fig.~\ref{fig: MPEC fail1} can occur without extra errors.

An $FF$ on the data towards the end of the circuit will be passed out of the EC box.

\begin{figure}[h!]
\begin{center}
\resizebox{80mm}{!}{\includegraphics{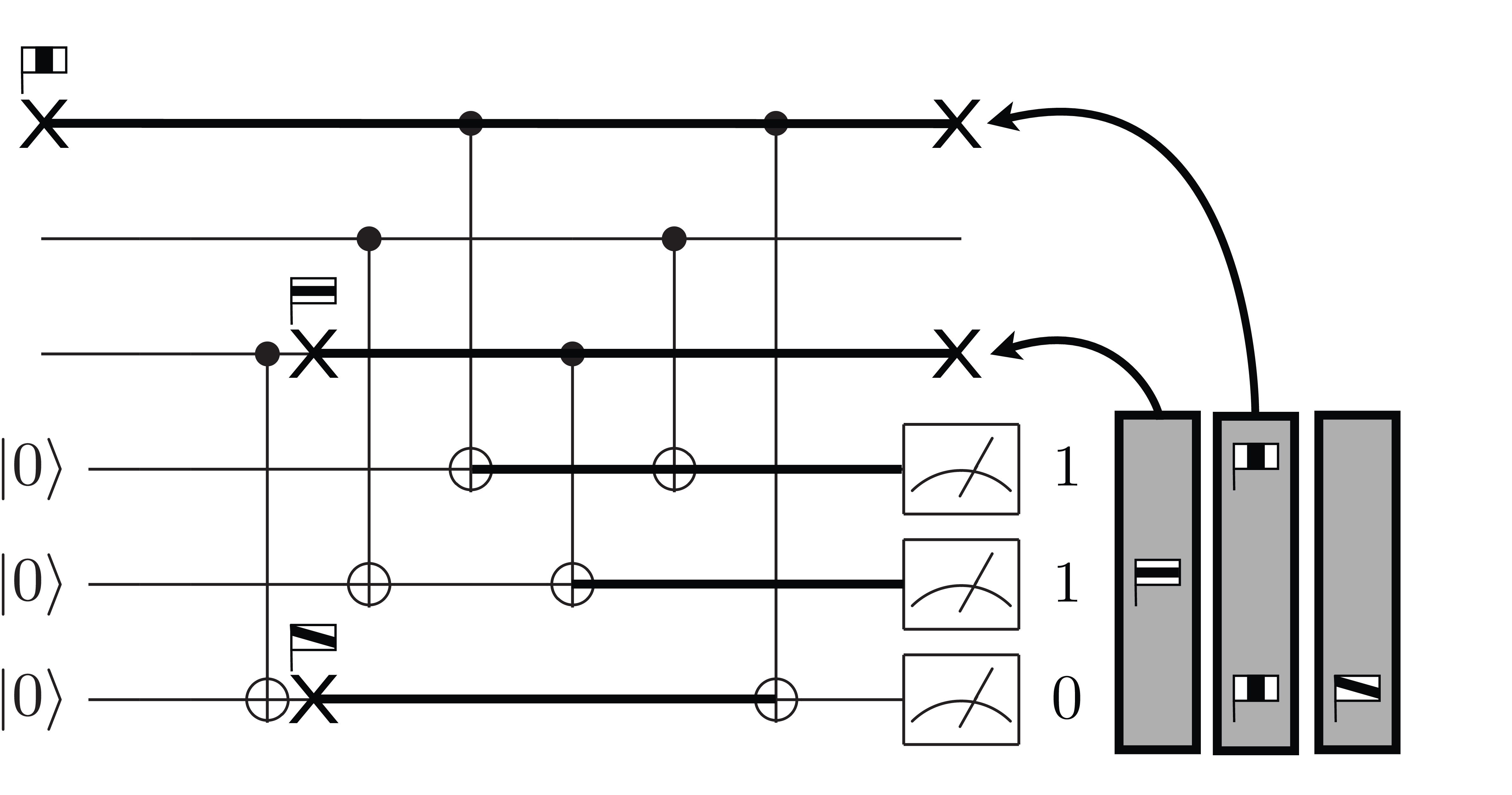}}
\end{center}
\vspace*{-10pt}
\caption{Incoming $FF$ with $FF+FF$ pair from a CNOT failure is successfully diagnosed and corrected with a three-flag match.}
\label{fig: MPEC2}
\end{figure}

\begin{figure}[h!]
\begin{center}
\resizebox{80mm}{!}{\includegraphics{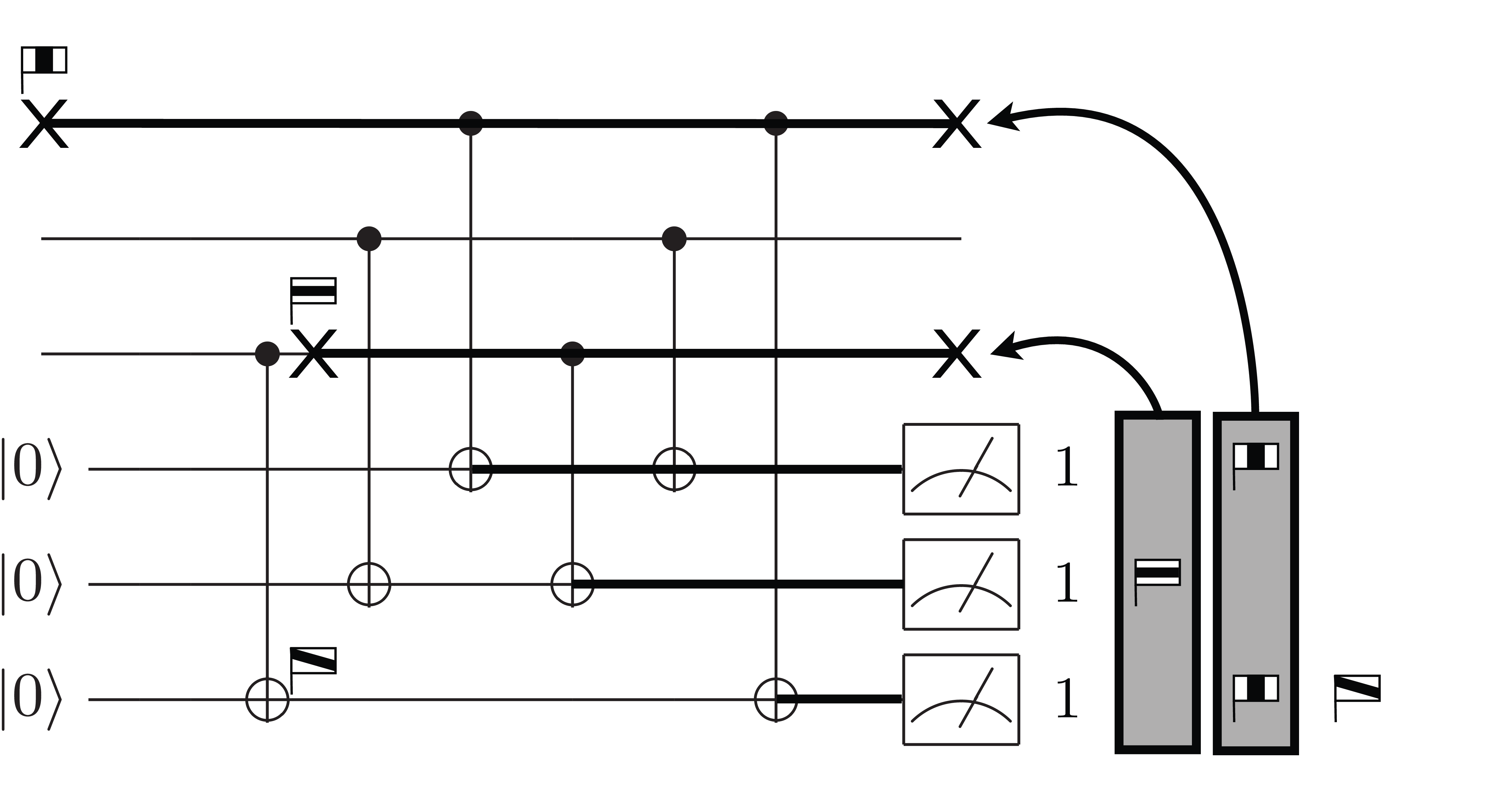}}
\end{center}
\vspace*{-10pt}
\caption{Incoming $FF$ with $FF+FS$ pair from a CNOT failure is successfully diagnosed and corrected. The horizontal and vertical flags together match the syndrome.}
\label{fig: MPEC3}
\end{figure}

\subsection{Cases with $a = 1, b \le 3$}
With $a = 1$ there must be an incoming $FS$ on one of the data qubits. Three physical errors inside the EC box is still only enough to cause failure at one location, but this time there may be a $UF$.

If flag from the $FS$ matches the syndrome, this will result in successful correction. There can be a $UF$ that produces a syndrome with no flag match. When there is no flag match MPEC falls back to the conventional Bacon-Shor error correction rules for this circuit. This will either result in a successful correction of the failure (if it occurs before or during the first set of CNOTs) or a parity mismatch on the syndrome and hence no correction. With $b=3$ it is acceptable to let a $UF$ or $FF$ leave the EC box, eg. Fig.~\ref{fig: MPEC4}.

\begin{figure}[h!]
\begin{center}
\resizebox{80mm}{!}{\includegraphics{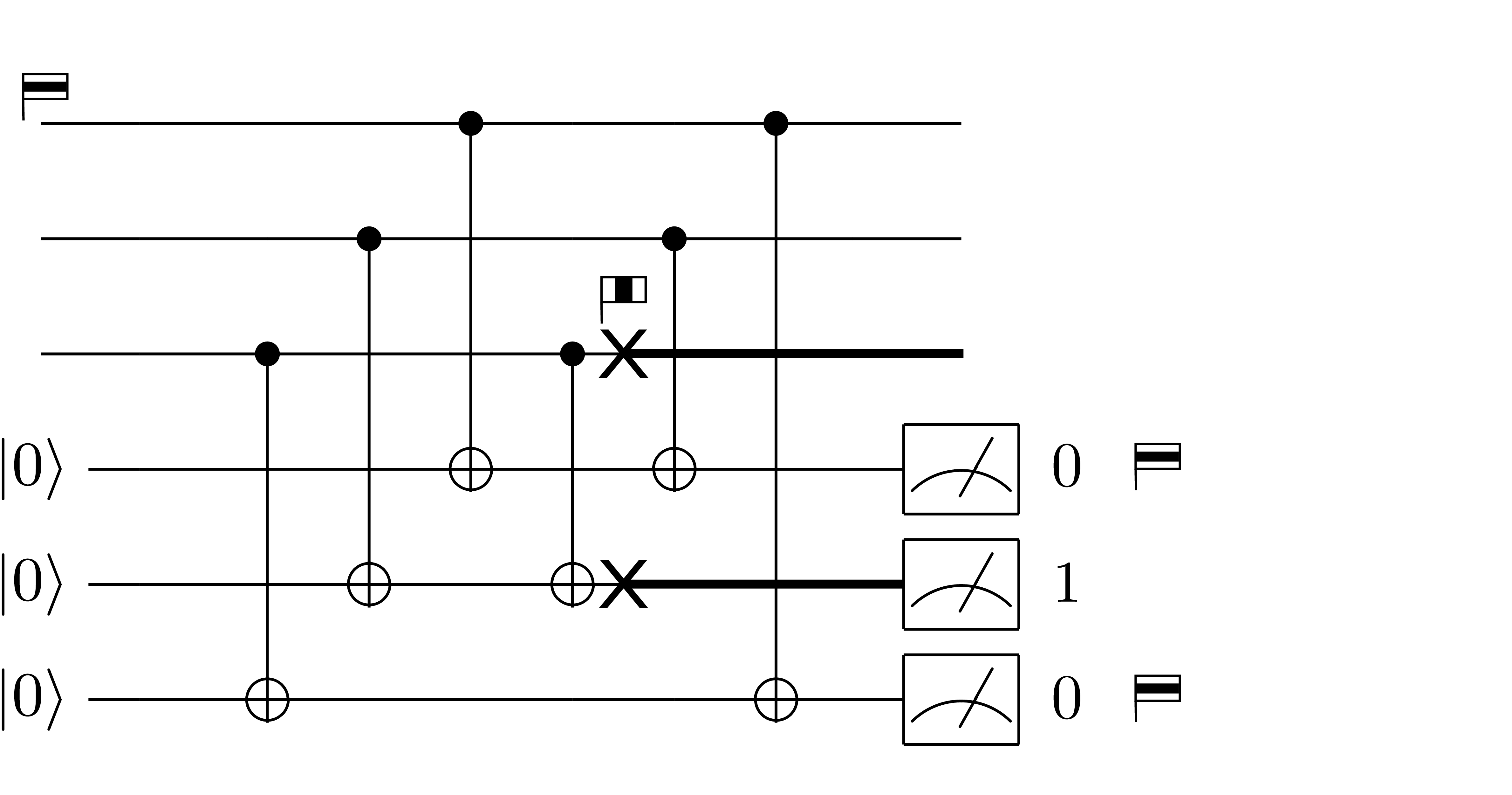}}
\end{center}
\vspace*{-10pt}
\caption{Incoming $FS$ with $FF+UF$ pair from a CNOT failure produces a syndrome which cannot be matched with any combination of the two flags and so the syndrome must be interpreted without the use of the flags. In this case, there is a parity mismatch in the syndrome and so the standard procedure is that no correction should be applied. The error on the third data qubit is let through. The vertical flag does not appear on the ancillary qubits and so it is not removed from that data. Therefore there is an $FF$ leaving this EC box.}
\label{fig: MPEC4}
\end{figure}

\subsection{Cases with $a = 0, b \le 4$}
With $b = 4$ there can either be two separate flagged failure locations, or one $UF$ location and one $FS$ location. As we have already seen, the $UF$ cases are all work in the same way as they do for conventional error correction. The worst case is shown in Fig.~\ref{fig: MPEC5}, where two pairs of physical errors case two $FF$ pairs. When this happens there will be a two flag match which will not correct all of the errors but will not cause logical failure. Since $b = 4$ it is acceptable for the EC box to let through a single $UF$ or even two $FF$s.

With $b = 4$ there are still not enough errors to affect two data qubits with errors and have a flag on the third line (Fig.~\ref{fig: MPEC fail1}) or to have a $UF$ on one data qubit that matches misleading flags on the other two qubit lines (Fig.~\ref{fig: MPEC fail2}).
\begin{figure}[h!]
\begin{center}
\resizebox{80mm}{!}{\includegraphics{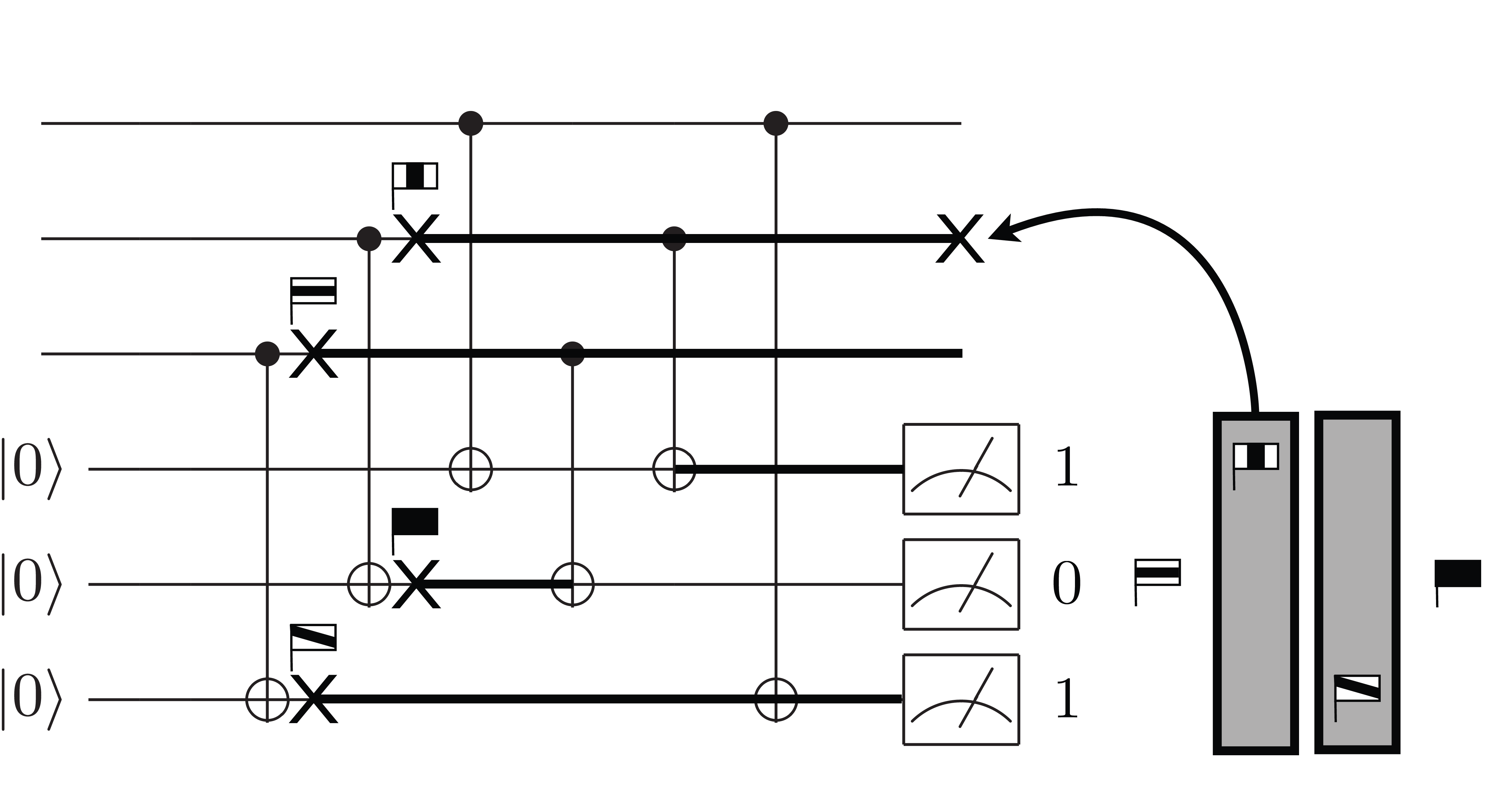}}
\end{center}
\vspace*{-10pt}
\caption{Two $FF+FF$ pairs from CNOT failures result in a two-flag match which corrects only one of the errors on the data qubits. Since all flags that appear on the ancillary qubits are cleared after error correction there will be a $UF$ leaving this EC box.}
\label{fig: MPEC5}
\end{figure}


\section{}
Having demonstrated in Appendix \ref{sec: Ap. A} that Eq.~\ref{eq: FT rule} is satisfied, in this section we aim to prove that improved failure rate scaling for MPEC extended rectangles found in Section \ref{sec: simulation} carries over to arbitrary length logical computation. For simplicity we consider only identity gates at the logical level, but it is not hard to extend this proof to an arbitrary logical circuit given standard requirements of fault tolerant circuits, which are already satisfied by the Bacon-Shor code (eg. transversal gates).

The probability that a logical circuit will fail can be bounded by the probability that each level-2 extended rectangle in the circuit has at most four physical errors.

An error correction box is defined to have succeeded if an errorless error correction cycle applied to its output state would produce the correct state. ie. if an EC box succeeded then any errors remaining on the output can be corrected. Let $a_n$ be the physical error weight of the input state to the $n^{th}$ EC box, and $b_n$ be the number of physical errors occurring inside the $n^{th}$ EC box. Let $S_n$ denote the success of the $n^{th}$ EC box. Then the success condition can be written as
\begin{equation}
	(b_{n+1} = 0 \Rightarrow a_{n+2} = 0)\Rightarrow S_n.
	\label{eq: success condition}
\end{equation}

\begin{theorem}

Given that there are no errors on the initial state and that all level-2 extended rectangles have at most four physical errors, all extended rectangles will succeed.
\end{theorem}

\begin{proof}

No initial errors:
\begin{equation}
	a_0 = 0.
	\label{eq: initial}
\end{equation}
All extended rectangles have at most four errors, that is, the sum of the errors in any two adjacent EC boxes is no more than 4:
\begin{equation}
	b_n+b_{n+1} \leq 4.
	\label{eq: ex-rects}
\end{equation}
MPEC is fault tolerant, as described in Section \ref{sec: fault tolerance} and demonstrated in Appendix \ref{sec: Ap. A}:
\begin{equation}
	(a_n+b_n \leq 4) \Rightarrow (a_{n+1}\leq b_n).
	\label{eq: FT logic}
\end{equation}

The first extended rectangle will have no more than four errors and there are no errors in the initial state (Eq.~\ref{eq: initial} and Eq.~\ref{eq: ex-rects}), therefore
\begin{equation}
	a_0 + b_0 \leq 4.
	\label{eq: a0b0}
\end{equation}

Suppose $a_n+b_n \leq 4$, then $a_{n+1}+b_{n+1} \leq b_n + b_{n+1}$
by the fault tolerance rule, Eq.~\ref{eq: FT logic} (with $b_{n+1}$ added to both sides). Therefore, using Eq.~\ref{eq: ex-rects} we have
\begin{equation}
	(a_n+b_n \leq 4) \Rightarrow (a_{n+1}+b_{n+1}) \leq 4.
	\label{eq: induction rule}
\end{equation}
With Eq.~\ref{eq: a0b0} and Eq.~\ref{eq: induction rule},
\begin{equation}
	a_n+b_n \leq 4
\end{equation} is true by induction. With Eq.~\ref{eq: FT logic} we then have
\begin{equation}
	a_{n+1} < b_{n}
	\label{eq: a<b}
\end{equation}
for all $n$. Therefore
\begin{equation}
(b_{n+1} = 0) \Rightarrow (a_{n+2} = 0),
\end{equation}
which with Eq.~\ref{eq: success condition} means success for all EC boxes. $\Box$

\end{proof}

\end{document}